\documentclass[12pt,leqno]{article}

\input{style}
\numberwithin{equation}{section}
\begin{document}
%%%%%%%%%%%%%%%%%%%%%%%%%%%%%%%%%%%%%%%%%%%%%%%%%%%%%%%%%%%%%%%%%%

\title{A model for a large investor trading at market indifference
  prices. I: single-period case.}

\author{P. Bank \thanks{This research was supported in part by the
    National
    Science Foundation under grant DMS-0505021.}\\
  Technische Universit{\"a}t Berlin\\
  and Quantitative Products Laboratory \\
  Institut f{\"u}r Mathematik\\
  Stra{\ss}e des 17. Juni 135, 10623 Berlin, Germany \\
  (bank@math.tu-berlin.de) \and D. Kramkov \thanks{The author also
    holds a part-time position at the University of Oxford. This
    research was supported in part by the National Science Foundation
    under grant DMS-0505414 and by the Oxford-Man Institute for
    Quantitative Finance at the University of Oxford.} \\
  Carnegie Mellon University,\\
  Department of  Mathematical Sciences,\\
  5000 Forbes Avenue, Pittsburgh, PA, 15213-3890, US \\
  (kramkov@cmu.edu)} \date{\today}

\maketitle
\begin{abstract}
  We develop a single-period model for a large economic agent who
  trades with market makers at their utility indifference prices.  We
  compute the sensitivities of these market indifference prices with
  respect to the size of the investor's order. It turns out that the
  price impact of an order is determined both by the market makers'
  joint risk tolerance and by the variation of individual risk
  tolerances. On a technical level, a key role in our analysis is
  played by a pair of conjugate saddle functions associated with the
  description of Pareto optimal allocations in terms of the aggregate
  utility function.
\end{abstract}

\begin{description}
\item[Keywords:] aggregate utility function, Bertrand competition,
  demand pressure, equilibrium, large investor, liquidity, Pareto
  allocation, price impact, risk-tolerance, utility indifference
  prices.
\item[JEL Classification:] G11, G12, G13, C61.
\item[AMS Subject Classification (2010):] 52A41, 60G60, 91G10, 91G20.
\end{description}

% \tableofcontents

\section{Introduction}
\label{sec:introduction}

Most models studied in Mathematical Finance specify price dynamics
exogenously, e.g., by some semi-martingale which satisfies certain
economically sensible conditions, most notably the absence of
arbitrage opportunities, and whose behavior is not affected by an
investor's trading strategy.  While being linear with respect to the
order flow and therefore convenient to analyze, such models inevitably
ignore the fundamental economic principle that prices are formed by
demand and supply. This idealization is justified from a practical
viewpoint as long as the trading volume remains small enough to be
easily covered by market liquidity. But when transactions involve a
sizable part of `the market' for a security, a model for their impact
on asset prices is clearly called for.

It is the purpose of this paper and its companion~\cite{BankKram:11b}
to systematically develop such a price impact model in, respectively,
a one-period and a continuous-time framework.  Of course, price impact
models have been proposed before with different goals in mind and we
refer to the survey by \citet{GokayRochSoner:11} for a more detailed
overview. For instance, \citet{AlmgChr:01} and \citet{SchSchon:09}
seek to find optimal execution strategies for selling a large quantity
of shares, while \citet{FreyStr:97, PlatSch:98, PapanSirc:98} and
\citet{CvitMa:96} as well as \citet{CetinJarrProt:04} focus on option
pricing in illiquid markets. The impact of transactions on asset
prices in their models is typically specified explicitly, in a
mathematically convenient form, but not derived systematically from
equilibrium-based considerations of economic theory.

There is a vast literature on price formation in economics and we
refer to \citet{AmihudMendPeder:05}, \citet{BiaisGlosSpatt:05} and
\citet{Hara:95} for surveys of financial models with price impact. As
in the model of \citet{Stoll:78}, also in our framework prices are
quoted by some market makers and the price impact is due to their risk
aversion against growing inventories of risky securities. In line with
the utility based valuation methods common in mathematical finance, we
add the feature that our market makers can hedge against the incurred
risk by trading freely among themselves. We argue that, due to
competition, this hedging opportunity leads to more favorable quotes
for the proposed transaction and in fact eventually leaves every
market maker indifferent between the pre- and the post-transaction
allocation of wealth.
 
We develop a mathematical framework for the analysis of these market
indifference prices. The key tool in our study is the convex duality
theory for saddle functions as presented by \citet{Rock:70}. This
allows us, for instance, to compute the asymptotic expansion of market
indifference prices with respect to the order size. It turns out that
these expansions have a component which can be described in terms of a
representative agent's risk aversion and another component which
reflects the diversity of risk aversions among the market makers.

The paper is structured as follows. In Section~\ref{sec:model} we give
the mathematical specification of our model and continue the
discussion of its economic features and some of the pertaining
literature.  In this section we also obtain a basic result on the
existence and uniqueness of market indifference prices. A detailed
study of these prices including the computation of their sensitivities
with respect the order size constitutes our main result and is
presented in Section~\ref{sec:analysis}. This analysis relies on a
number of technical facts related to the classical parametrization of
Pareto optimal allocations in terms of the aggregate utility
function. They are collected in Section~\ref{sec:Pareto}.

\section{Model}
\label{sec:model}

We consider a single-period financial model with initial time~$0$ and
maturity~$1$ where $M\in \braces{1,2,\dots}$ market makers quote
prices for $J\in \braces{1,2,\dots}$ traded assets. These assets are
specified through their payoffs $\psi=(\psi^1,\ldots,\psi^J)$ at time
$1$ which are random variables defined on a complete probability space
$(\Omega, \mathcal{F}, \mathbb{P})$.  As usual, we identify random
variables differing on a set of measure zero and use the notations
$\mathbf{L}^0(\mathbf{R}^d)$, for the metric space of such equivalence
classes with values in $\mathbf{R}^d$ and convergence in probability,
and $\mathbf{L}^p(\mathbf{R}^d)$, $p\geq 1$, for the Banach space of
$p$-integrable random variables.

The way the market makers serve the incoming orders crucially depends
on their attitude toward risk, which we model in the classical
framework of expected utility. Thus, we interpret the probability
measure $\mathbb{P}$ as a description of the common beliefs of our
market makers (same for all, for simplicity) and denote by
$u_m=(u_m(x))_{x\in \mathbf{R}}$ market maker $m$'s utility function
for terminal wealth.
  
\begin{Assumption}
  \label{as:1} Each $u_m = u_m(x)$, $m=1,\dots,M$, is a strictly
  concave, strictly increasing, continuously differentiable, and
  bounded above function on the real line $\mathbf{R}$ satisfying
  \begin{equation}
    \label{eq:1}
    \lim_{x\to \infty} u_m(x) = 0.
  \end{equation}
\end{Assumption}

The normalization to zero in~\eqref{eq:1} is added only for notational
convenience. Many of our results will be derived under the following
additional condition on the utility functions, which, in particular,
implies their boundedness from above.

\begin{Assumption}
  \label{as:2}
  Each utility function $u_m = u_m(x)$, $m=1,\dots,M$, is twice
  continuously differentiable and its absolute risk aversion
  coefficient is bounded away from zero and infinity, that is, for
  some $c>0$,
  \begin{equation}
    \label{eq:2} 
    \frac1c \leq a_m(x) \set -\frac{u_m''(x)}{u_m'(x)}
    \leq c, \quad x \in \mathbf{R}.
  \end{equation}
\end{Assumption}

The prices quoted by the market makers are also influenced by their
initial endowments $\alpha_0=(\alpha^m_0)_{m=1,\dots,M} \in
\mathbf{L}^0(\mathbf{R}^M)$, where $\alpha^m_0$ is an
$\mathcal{F}$-measurable random variable describing the terminal
wealth of the $m$th market maker if no orders have to be filled. We
assume that the initial allocation $\alpha_0$ is \emph{Pareto
  optimal}, that is, there is no strictly better re-allocation of the
same resources in the sense of the following well-known

\begin{Definition}
  \label{def:1}
  A vector of $\mathcal{F}$-measurable random variables
  $\alpha=(\alpha^m)_{m=1,\dots,M}$ is called a \emph{Pareto optimal
    allocation} if
  \begin{equation}
    \label{eq:3}
    \mathbb{E}[|u_m(\alpha^m)|] < \infty, \quad
    m=1,\dots,M,
  \end{equation}
  and if there is no other allocation $\beta \in
  \mathbf{L}^0(\mathbf{R}^M)$ with the same total endowment,
  \begin{displaymath}
    \sum_{m=1}^M \beta^m=\sum_{m=1}^M \alpha^m,
  \end{displaymath}
  which leaves all market makers not worse and at least one of them
  better off in the sense that
  \begin{displaymath}
    \mathbb{E}[{u_m(\beta^m)}] \geq \mathbb{E}[{u_m(\alpha^m)}]
    \mtext{for all} m=1,\dots,M,
  \end{displaymath}
  and
  \begin{displaymath}
    \mathbb{E}[{u_m(\beta^m)}] > \mathbb{E}[{u_m(\alpha^m)}]
    \mtext{for some} m \in \braces{1,\dots,M}. 
  \end{displaymath}
\end{Definition}

Suppose that at time $0$ the market makers fill an order from a large
investor for the quantity $q\in \mathbf{R}^J$ of the assets $\psi\in
\mathbf{L}^0(\mathbf{R}^J)$.  As a result, the market makers' total
endowment changes from $\Sigma_0 \set \sum_{m=1}^M \alpha^m_0$ to
\begin{displaymath}
  \Sigma(x,q) \set \Sigma_0 + x + \ip{q}{\psi} =
  \Sigma_0 + x + \sum_{j=1}^J q^j \psi^j,
\end{displaymath} 
where $\ip{\cdot}{\cdot}$ denotes the Euclidean scalar product and
$x\in \mathbf{R}$ and $q\in\mathbf{R}^J$ stand for, respectively, the
cash amount and the number of assets acquired by the market
makers. Our model will assume that $\Sigma(x,q)$ is re-allocated among
the market makers in the form of a Pareto optimal allocation. For this
to be possible we have to impose

\begin{Assumption}
  \label{as:3}
  For any $x\in \mathbf{R}$ and $q\in \mathbf{R}^J$ there is an
  allocation $\beta \in \mathbf{L}^0(\mathbf{R}^M)$ with the total
  endowment $\Sigma(x,q)$ such that
  \begin{equation}
    \label{eq:4}
    \mathbb{E}[u_m(\beta^m)] > - \infty, \quad m=1,\dots,M. 
  \end{equation}
\end{Assumption}

\begin{Remark}
  \label{rem:1}
  Under Assumption~\ref{as:1}, Lemma~\ref{lem:1} below shows that
  Assumption~\ref{as:3} is equivalent to the finiteness of the
  aggregate expected utility. Under Assumptions~\ref{as:1}
  and~\ref{as:2}, Lemma~\ref{lem:3} proves that Assumption~\ref{as:3}
  is also equivalent to the existence of all exponential moments for
  $\psi$ under the pricing measure $\mathbb{Q}_0$ associated with the
  initial Pareto optimal allocation $\alpha_0$.
\end{Remark}

\subsection{Market indifference prices}
\label{sec:exist-mark-indiff}

Theorem~\ref{th:1} below identifies uniquely both the transaction
price $x=x(q)$ associated with an order $q$ and the market makers'
post-transaction allocation of wealth $\alpha_1=\alpha_1(q)$ under the
following two conditions:
\begin{description}
\item[Pareto optimality:] The random endowment $\Sigma(x,q)$ is
  redistributed between the market makers to form a new \emph{Pareto
    optimal} allocation $\alpha_1$.
\item[Utility indifference:] The market makers' expected utilities
  \emph{do not change}:
  \begin{equation}
    \label{eq:5}
    \mathbb{E}[u_m(\alpha^m_1)] =  \mathbb{E}[u_m(\alpha^m_0)],  \quad
    m=1,\dots,M. 
  \end{equation}
\end{description}
We postpone the discussion of the economic motivation behind these two
conditions until Section~\ref{sec:model-discussion}. By analogy with
the popular {utility-based valuation} method in mathematical finance,
we call $x=x(q)$ the \emph{market indifference price} of the
order~$q$.

\begin{Theorem}
  \label{th:1}
  Under Assumptions~\ref{as:1} and~\ref{as:3}, every position $q\in
  \mathbf{R}^J$ yields a unique cash amount $x=x(q)$ and a unique
  Pareto optimal allocation $\alpha_1=\alpha_1(q)$ of $\Sigma(x,q)$
  preserving the market makers' expected utilities in the sense of
  \eqref{eq:5}.
\end{Theorem}

\begin{proof}
  For a real number $y$ denote by $\mathcal{B}(y)$ the family of
  allocations $\beta = (\beta^m)_{m=1,\dots,M}$ with total endowment
  less than $\Sigma(y,q)$ and such that
  \begin{displaymath}
    \mathbb{E}[{u_m(\beta^m)}] \geq
    \mathbb{E}[{u_m(\alpha_0^m)}], \quad m=1,\dots,M. 
  \end{displaymath}
  By Assumptions~\ref{as:1} and~\ref{as:3}, this set is non-empty for
  sufficiently large $y$ and, by the concavity of utility functions,
  is a convex subset of $\mathbf{L}^0(\mathbf{R}^M)$.  Denote
  \begin{displaymath}
    \widehat y \set \inf \descr{y\in \mathbf{R}}{\mathcal{B}(y)\not=
      \emptyset},
  \end{displaymath}
  let $(y_n)_{n\geq 1}$ be a strictly decreasing sequence of real
  numbers converging to $\widehat y$, and arbitrarily choose $\beta_n
  \in \mathcal{B}(y_n)$, $n\geq 1$.

  From Assumption~\ref{as:1} we deduce the existence of $c>0$ such
  that, for $m=1,\dots,M$,
  \begin{displaymath} y^- \leq c(-u_m(y)), \quad y\in \mathbf{R},
  \end{displaymath} where $y^- \set \max(0, -y)$; for example, we can
  take
  \begin{displaymath} c = 1/\min_{m=1,\dots,M}u'_m(0).
  \end{displaymath} It follows that
  \begin{displaymath} \mathbb{E}[(\beta_n^m)^-)] \leq c
    \mathbb{E}[(-u_m(\beta^m_n))] \leq c
    \mathbb{E}[(-u_m(\alpha_0^m))] < \infty, \quad n\geq 1,
  \end{displaymath} and, therefore, the sequence $((\beta_n)^-)_{n\geq
    1}$ is bounded in $\mathbf{L}^1(\mathbf{R}^M)$. Since, in addition,
  \begin{displaymath} \sum_{m=1}^M \beta_n^m \leq \Sigma(y_n,q) \leq
    \Sigma(y_1,q),
  \end{displaymath} the family of all possible convex combinations of
  $(\beta_n)_{n\geq 1}$ is bounded in $\mathbf{L}^0(\mathbf{R}^M)$.

  By Lemma A1.1 in \cite{DelbSch:94} we can then choose convex
  combinations $\zeta_n$ of $(\beta_k)_{k\geq n}$, $n\geq 1$, which
  converge almost surely to a random variable $\zeta \in
  \mathbf{L}^0(\mathbf{R}^M)$.  Clearly,
  \begin{equation}
    \label{eq:6} 
    \sum_{m=1}^M \zeta^m \leq \Sigma(\widehat y, q).
  \end{equation} 
  Since the utility functions are bounded from above, Fatou's lemma
  yields
  \begin{equation}
    \label{eq:7}
    \mathbb{E}[u_m(\zeta^m)] \geq \limsup_{n\to\infty}
    \mathbb{E}[u_m(\zeta^m_n)] \geq \mathbb{E}[u_m(\alpha_0^m)], 
  \end{equation}
  where the second estimate holds because $\zeta_n\in
  \mathcal{B}(y_n)$ by the convexity of $\mathcal{B}(y_n)$. It follows
  that $\zeta \in \mathcal{B}(\widehat y)$.  The minimality of
  $\widehat y$ then immediately implies the equalities in \eqref{eq:6}
  and \eqref{eq:7} and the Pareto optimality of $\zeta$. Hence, we can
  select $x = \widehat y$ and $\alpha_1 = \zeta$, thus proving their
  existence.

  Finally, the uniqueness of $x$ and $\alpha_1$ follows from the
  strict concavity of utility functions.
\end{proof}

The preceding proof yields immediately that the market indifference
price $x(q)$ is minimal in the following sense:

\begin{Corollary}
  \label{cor:1}
  Under Assumptions~\ref{as:1} and~\ref{as:3}, the market indifference
  price $x(q)$ associated with an order $q \in \mathbf{R}^J$ is the
  minimal amount $x \in \mathbf{R}$ for which there exists an
  allocation $\beta \in \mathbf{L}^0(\mathbf{R}^M)$ with the total
  wealth $\Sigma(x,q)$ such that
  \begin{displaymath}
    \mathbb{E}[{u_m(\beta^m)}] \geq
    \mathbb{E}[{u_m(\alpha_0^m)}], \quad m=1,\dots,M.
  \end{displaymath}
\end{Corollary}

In Section~\ref{sec:analysis} we show in detail how the market
indifference price $\map{x=x(q)}{\mathbf{R}^J}{\mathbf{R}}$ and the
Pareto allocations
$\map{\alpha_1=\alpha_1(q)}{\mathbf{R}^J}{\mathbf{L}^0(\mathbf{R}^M)}$
determined by Theorem \ref{th:1} depend on the order $q$ of the large
investor.

\subsection{Economic considerations}
\label{sec:model-discussion}

Let us next discuss the scope and limitations of our model as well as
its relation to other models proposed in the literature. The book by
\citet{Hara:95} as well as the surveys by \citet{AmihudMendPeder:05}
and by \citet{BiaisGlosSpatt:05} provide an extensive overview of the
economic background.

Like in the work of \citet{Stoll:78} and \citet{HoStoll:81} also in
our model the market makers' quotes depend on the inventory they hold.
The first distinctive feature of our setup is the assumption that the
market makers share their risk in a Pareto optimal way. Since Pareto
optimal allocations of wealth can be achieved, essentially, only in
complete markets, we are implicitly postulating that, besides the
publicly traded assets $\psi$, the market makers can also trade any
arbitrarily structured product among themselves. In other words, they
have a complete `over the counter market' at their disposal. It would
be interesting to obtain conditions under which this auxiliary market
can be dispensed with and the desired completeness property is
achieved \emph{endogenously} by trading only the primary securities
$\psi$. In this respect, we mention the works by \citet{AnderRaim:08,
  HugMalTrub:12, RiedHerz:13, Kram:13b} on the existence of an
endogenously complete Arrow-Radner equilibrium.

As a second key condition we assume that the market makers are
indifferent between their pre- and post-transaction allocations of
wealth. Given the existence of a complete market at their disposal,
this can be viewed as a consequence of a Bertrand-style competition
among them. Indeed, on one hand Corollary~\ref{cor:1} shows that there
is no way to trade $q$ securities at a transaction price less than
$x(q)$ if no market maker is to loose in terms of expected
utility. On the other hand, if one of the market makers could fill the
order $q$ at a price $x>x(q)$ she could subsequently use a small part
of the difference $x-x(q)$ to offer attractive deals to her fellow
market makers on the over the counter market, e.g., by topping up the
utility preserving allocation $\alpha_1(q)$ of Theorem~\ref{th:1}.
Plainly, all market makers would like to proceed this way. Following
the usual Bertrand competition argument, this leads to a limiting
quote of the market indifference price $x(q)$ and all the market
makers end up being indifferent between the pre- and the
post-transaction allocation of wealth.  It may be interesting to note
that the market indifference price $x(q)$ is smaller than every of the
market makers' reservation prices $x_m(q)$ defined by
\begin{displaymath}
  \mathbb{E}[{u_m(\alpha_0^m+x_m(q)+\ip{q}{\psi}})]
  = \mathbb{E}[{u_m(\alpha_0^m)}], \quad m=1,\dots,M.
\end{displaymath} 
This follows directly from Corollary~\ref{cor:1} and is due to the
market makers' ability to improve their positions by trading with each
other.

There are, of course, several idealizations in our model. For
instance, it assumes that all trades between the market markers are
completed instantly after the transaction with the large investor. But
in real markets it may take time to find counterparties for a hedge
and it may be impossible to negotiate with them the best possible
deals. Hence, to compensate for risks emerging from these frictions,
real market makers will ask for a risk premium and for compensation of
their operational costs both of which are not included in our
transaction price; see \citet{GrosMill:88, DuffGarlPeder:05, Long:01}
and the references therein.  Another idealization is that the market
makers do not act strategically. In particular, they do not account
for adverse selection effects as studied, e.g., by \citet{Kyle:85,
  Back:92, GlosMilg:85, BiaisMartRoch:00}.

As pointed out by \citet{BernHugh:07} the utility indifference may not
hold when, by contrast to our setting, orders can be split \emph{and}
market makers quote price schedules which have to be honored
independently from the fellow market makers and when hedging
opportunities are not taken into account.  Our model is different from
the one in \citet{GarlPederPotes:09} where a representative market
maker maximizes utility from consumption and where prices are
determined so that her optimal position matches some exogenously given
demand for marketed assets. Our market makers also do not exert market
power as investigated, e.g., in the work of \citet{Weretka:11}. While
it would certainly be desirable to account for all these various
aspects of market microstructure, this is clearly beyond the scope of
a single paper.

\section{Parametrization of Pareto allocations}
\label{sec:Pareto}

The analysis of market indifference prices in
Section~\ref{sec:analysis} relies on the parametrization of Pareto
allocations stated, in our setting, in Theorem~\ref{th:2} below. This
parametrization is well-known; see, e.g., \citet{DanaLeVan:96} for a
similar result.

Recall that the \emph{aggregate utility function} is defined as
\begin{equation}
  \label{eq:8}
  r(v,x) \set \sup_{x^1 + \dots + x^M = x} \sum_{m=1}^M v^m u_m(x^m), \quad
  v\in (0,\infty)^M, \; x\in \mathbf{R}.
\end{equation}
One can show, see, e.g., Theorem~3.1 of \cite{BankKram:13b}, that under
Assumption~\ref{as:1} $r = r(v,x)$ is a continuously differentiable
saddle function, strictly convex and decreasing in $v \in
(0,\infty)^M$ and strictly concave and increasing in $x \in
\mathbf{R}$. Moreover, for any $(v,x)\in (0,\infty)^M \times
\mathbf{R}$, the supremum in~\eqref{eq:8} is attained at the vector
$\widehat{x} = (\widehat{x}^m)_{m=1,\dots,M}\in \mathbf{R}^M$ uniquely
determined by
\begin{align}
  \label{eq:9}
  v^m u'_m(\widehat{x}^m) &= \frac{\partial r}{\partial x} (v,x), \\
  \intertext{or, equivalently, by}
  \label{eq:10}
  u_m(\widehat{x}^m) & = \frac{\partial r}{\partial v^m} (v,x), \quad
  m=1, \dots, M.
\end{align}

\begin{Theorem}
  \label{th:2}
  Under Assumption~\ref{as:1}, the following statements are equivalent
  for every $\Sigma \in \mathbf{L}^0(\mathbf{R})$:
  \begin{enumerate}
  \item \label{item:1} The allocation $\alpha =
    (\alpha^m)_{m=1,\dots,M}$ of $\Sigma$ is Pareto optimal.
  \item \label{item:2} The random variables $\alpha =
    (\alpha^m)_{m=1,\dots,M}$ satisfy the integrability condition
    \eqref{eq:3} and there is a a vector $v \in (0,\infty)^M$ such
    that
    \begin{equation}
      \label{eq:11}
      v^m u_m'(\alpha^m) = \frac{\partial r}{\partial x}
      (v, \Sigma), \quad m=1,\dots,M.
    \end{equation}
  \end{enumerate}
  Moreover, such a vector $v$ is defined uniquely up to a scalar
  multiple.
\end{Theorem}

\begin{proof} Denote by $\mathcal{B}=\mathcal{B}(\Sigma)$ the family
  of allocations $\beta\in\mathbf{L}^0(\mathbf{R}^M)$ with the total
  endowment $\sum_{m=1}^M \beta^m = \Sigma$ such that
  \begin{displaymath}
    \mathbb{E}[u^m(\beta^m)] > -\infty, \quad m=1,\dots,M. 
  \end{displaymath}

  \ref{item:1} $\implies$ \ref{item:2}: The integrability
  condition~\eqref{eq:3} for $\alpha$ holds by the definition of a
  Pareto optimal allocation. By the concavity of the utility
  functions, the set
  \begin{displaymath}
    C\set \descr{z\in \mathbf{R}^M}{z^m \leq \mathbb{E}[u_m(\beta^m)], \;
      m=1,\dots,M, \mtext{for some} \beta\in \mathcal{B}}
  \end{displaymath}
  is convex and, by the Pareto optimality of $\alpha$, the point
  \begin{displaymath}
    \widehat z^m \set \mathbb{E}[u_m(\alpha^m)], \quad m=1\dots, M,
  \end{displaymath}
  belongs to the boundary of $C$. Hence, there is a non-zero $v\in
  \mathbf{R}^M$ such that
  \begin{displaymath}
    \ip{v}{\widehat z} \geq \ip{v}{z}, \quad z\in C, 
  \end{displaymath}
  or, equivalently,
  \begin{displaymath}
    \mathbb{E}[\sum_{m=1}^M v^m u_m(\alpha^m)] = \sup_{\beta\in \mathcal{B}}
    \mathbb{E}[\sum_{m=1}^M v^m u_m(\beta^m)] =  \mathbb{E}[r(v,\Sigma)]. 
  \end{displaymath}
  As $v\not=0$, the properties of the utility functions in
  Assumption~\ref{as:1} imply that $v\in
  (0,\infty)^M$. By~\eqref{eq:9}, the fact that the upper bound above
  is attained at $\alpha$ is equivalent to~\eqref{eq:11}.

  2 $\implies$ 1: By the definition of $r=r(v,x)$, for every $\beta
  \in \mathcal{B}$
  \begin{align*}
    \sum_{m=1}^M v^m u_m(\beta^m) \leq r(v,\Sigma) = \sum_{m=1}^M v^m
    u_m(\alpha^m).
  \end{align*}
  Given the integrability requirement \eqref{eq:3}, this clearly
  implies the Pareto optimality of $\alpha$.

  Finally, we note that \eqref{eq:11} determines $v$ uniquely up to a
  scalar multiple.
\end{proof}

It is common to normalize the Pareto weight $v$ from
Theorem~\ref{th:2} by restricting it to the set
\begin{displaymath}
  \mathbf{S}^M \set \descr{w \in (0,1)^M}{\sum_{m=1}^M w^m = 1},
\end{displaymath}
the interior of the simplex; such a $v\in \mathbf{S}^M$ is defined
uniquely.

The following result allows us to parametrize the Pareto allocations
in our economy by the set
\begin{equation}
  \label{eq:12}
  \mathbf{A} \set (0,\infty)^M \times \mathbf{R} \times \mathbf{R}^J,
\end{equation}

\begin{Lemma}
  \label{lem:1}
  Let Assumption~\ref{as:1} hold. Then Assumption~\ref{as:3} is
  equivalent to
  \begin{equation}
    \label{eq:13}
    \mathbb{E}[r(v,\Sigma(x,q))] > -\infty, \quad (v,x,q)\in
    \mathbf{A}.
  \end{equation}

  In this case, for $a=(v,x,q)\in \mathbf{A}$, the random vector
  $\pi(a) \in \mathbf{L}^0(\mathbf{R}^M)$ defined by
  \begin{equation}
    \label{eq:14}
    v^m u'_m(\pi^m(a)) = \frac{\partial{r}}{\partial x}(v,\Sigma(x,q)),\quad
    m=1,\dots,M,  
  \end{equation}
  forms a Pareto optimal allocation. Conversely, for $(x,q)\in
  \mathbf{R}\times \mathbf{R}^J$, any Pareto allocation of the total
  endowment $\Sigma(x,q)$ is given by \eqref{eq:14} for some $v\in
  (0,\infty)^M$.
\end{Lemma}

\begin{proof} Under Assumption~\ref{as:3}, there is an allocation
  $\beta = (\beta^m)_{m=1,\dots,M}$ of $\Sigma(x,q)$
  satisfying~\eqref{eq:4}. Since,
  \begin{displaymath}
    \sum_{m=1}^M v^m u_m(\beta^m) \leq r(v,\Sigma(x,q)) \leq 0, 
  \end{displaymath} 
  we obtain \eqref{eq:13}.

  Assume now that~\eqref{eq:13} holds. By~\eqref{eq:9}, the condition
  \eqref{eq:14} is equivalent to
  \begin{displaymath}
    \sum_{m=1}^M v^m u_m(\pi^m(a)) = r(v,\Sigma(x,q)).   
  \end{displaymath}
  As $v^m>0$ and $u_m<0$, $m=1,\dots,M$, we deduce from \eqref{eq:13}
  that $u_m(\pi^m(a)) \in \mathbf{L}^1$, $m=1,\dots,M$, which, in
  particular, yields Assumption~\ref{as:3}. The Pareto optimality of
  $\pi(a)$ is now an immediate corollary of Theorem~\ref{th:2}.

  Finally, the fact that any Pareto allocation of $\Sigma(x,q)$ is
  given by \eqref{eq:14}, for some $v\in (0,\infty)^M$, follows from
  Theorem~\ref{th:2}.
\end{proof}

In view of the equivalence of~\eqref{eq:9} and~\eqref{eq:10}, the
Pareto allocation $\pi(a)$ can also be defined by
\begin{equation}
  \label{eq:15}
  u_m(\pi^m(a)) = \frac{\partial{r}}{\partial v^m}(v,\Sigma(x,q)),\quad
  m=1,\dots,M.
\end{equation}

A well-known interpretation of the condition~\eqref{eq:14} for the
Pareto allocation $\pi(a)$ is that our market makers agree on the
valuation of marginal trades. The existence of the corresponding
\emph{pricing measure} $\mathbb{Q}(a)$ is established in the following
lemma.

\begin{Lemma}
  \label{lem:2}
  Under Assumptions~\ref{as:1} and~\ref{as:3}, for any $a=(v,x,q)\in
  \mathbf{A}$, there is a probability measure $\mathbb{Q}(a)$ such
  that
  \begin{equation}
    \label{eq:16}
    \frac{d\mathbb{Q}(a)}{d\mathbb{P}} = \frac{\frac{\partial r}{\partial
        x}(v,\Sigma(x,q))}{\mathbb{E}[\frac{\partial r}{\partial
        x}(v,\Sigma(x,q))]} 
    = \frac{u'_m(\pi^m(a))}{\mathbb{E}[u'_m(\pi^m(a))]}, \quad m=1,\dots,
    M,
  \end{equation}
  where $\pi(a) = (\pi^m(a))_{m=1,\dots,M}$ is the Pareto optimal
  allocation defined in \eqref{eq:14}.
\end{Lemma}

\begin{proof}
  In view of Lemma~\ref{lem:1} it is sufficient to verify that
  \begin{displaymath}
    \mathbb{E}[\frac{\partial r}{\partial x}(v,\Sigma(x,q))] < \infty.  
  \end{displaymath}
  This follows from \eqref{eq:13} and the inequality
  \begin{displaymath}
    \frac{\partial r}{\partial x}(v,\Sigma(x,q)) \leq r(v,\Sigma(x,q)) -
    r(v,\Sigma(x-1,q)), 
  \end{displaymath}
  which holds by the concavity of $r(v,\cdot)$.
\end{proof}

We conclude this section with another equivalent re-formulation of
Assumption~\ref{as:3}, this time, in terms of an exponential moment
condition for $\psi$ under the initial pricing measure $\mathbb{Q}_0$.

\begin{Lemma}
  \label{lem:3}
  Under Assumptions~\ref{as:1} and~\ref{as:2} the pricing measure
  $\mathbb{Q}_0$ of the initial Pareto allocation $\alpha_0 =
  (\alpha_0^m)_{m=1,\dots,M}$ is well-defined by
  \begin{displaymath}
    \frac{d\mathbb{Q}_0}{d\mathbb{P}} =
    \frac{u'_m(\alpha^m_0)}{\mathbb{E}[u'_m(\alpha^m_0)]}, \quad
    m=1,\dots, M.  
  \end{displaymath}  
  and Assumption~\ref{as:3} is equivalent to the existence of all
  exponential moments of $\psi = (\psi^j)_{j=1,\dots,J}$ under
  $\mathbb{Q}_0$:
  \begin{equation}
    \label{eq:17}
    \mathbb{E}_{\mathbb{Q}_0}[e^{z \abs{\psi}}] < \infty, \quad z\in
    \mathbf{R}.  
  \end{equation}
\end{Lemma}

\begin{proof}
  Observe first that Assumptions~\ref{as:1} and~\ref{as:2} imply that
  \begin{equation}
    \label{eq:18}
    \frac1c \leq -\frac{u'_m(x)}{u_m(x)} \leq c, \quad x\in
    \mathbf{R}, \quad  m=1,\dots,M, 
  \end{equation}
  with the same constant $c>0$ as in~\eqref{eq:2}.  From this
  and~\eqref{eq:9} we deduce that
  \begin{displaymath}
    \frac{M}c \frac{\partial r}{\partial x}(v,x) \leq - r(v,x) \leq cM 
    \frac{\partial r}{\partial x}(v,x), \quad (v,x)\in
    (0,\infty)^M\times \mathbf{R},
  \end{displaymath}
  and, therefore, 
  \begin{equation}
    \label{eq:19}
    \begin{split}
      -r(v,x+y) &\geq \frac{M}{c} \frac{\partial r}{\partial x}(v,x)
      \exp\left(\frac1{cM}y^- -
        \frac{c}{M}y^+\right), \\
      -r(v,x+y) &\leq cM \frac{\partial r}{\partial x}(v,x) \exp\left(
        \frac{c}{M}y^- - \frac{1}{cM}y^+\right),
    \end{split}
  \end{equation}
  where $y^+\set \max(y,0)$ and $y^- \set \max(-y,0)$.

  For $\mathbb{Q}_0$ to be well-defined we have to verify that
  \begin{displaymath}
    \mathbb{E}[u'_m(\alpha^m_0)] < \infty, \quad m=1,\dots, M. 
  \end{displaymath}
  This follows from \eqref{eq:18} and the inequality
  \begin{displaymath}
    \mathbb{E}[u_m(\alpha^m_0)] >- \infty, \quad m=1,\dots, M,
  \end{displaymath}
  which holds by the definition of a Pareto optimal allocation.

  Let $v_0\in (0,\infty)^M$ denote a Pareto weight of $\alpha_0$ from
  Theorem~\ref{th:2}. The positive homogeneity and the monotonicity
  properties of $r(\cdot,x)$ imply the equivalence of \eqref{eq:13}
  and
  \begin{displaymath}
    \mathbb{E}[r(v_0,\Sigma(x,q))] > -\infty, \quad (x,q)\in
    \mathbf{R}\times \mathbf{R}^J,
  \end{displaymath}
  which by~\eqref{eq:19} and the definition of $\mathbb{Q}_0$ is, in
  turn, equivalent to~\eqref{eq:17}.
\end{proof}

\section{Analysis of market indifference prices}
\label{sec:analysis}

The construction of market indifference prices in Theorem~\ref{th:1}
does not directly allow for an analysis of how these prices actually
depend on the order the market makers' are assumed to fill. We thus
develop in Section~\ref{sec:quantitative_description} below a
quantitative description of these prices. Indeed, Theorem~\ref{th:3}
shows in particular that the market indifference price and the
associated Pareto allocation can be computed from the saddle conjugate
of the representative market maker's expected
utility. Theorem~\ref{th:4} then computes the second order derivatives
of market indifference prices and the sensitivity of their Pareto
allocations with respect to the order
size. Section~\ref{sec:asyptotic_expansion} then uses these to provide
asymptotic expansions of market indifference prices which reveal their
dependence on the market makers' individual and aggregate risk
tolerances.

\subsection{Quantitative description of market indifference prices}
\label{sec:quantitative_description}

A key role in our study of market indifference prices will be played
by the pair of conjugate saddle functions
\begin{align*}
  F_0(a) &\set \mathbb{E}[r(v,\Sigma(x,q))], \quad
  a=(v,x,q) \in \mathbf{A},   \\
  G_0(b) &\set \sup_{v\in (0,\infty)^M}\inf_{x\in
    \mathbf{R}}[\ip{v}{u} + xy - F_0(v,x,q)], \quad b = (u,y,q)\in
  \mathbf{B},
\end{align*}
where $r=r(v,x)$, the aggregate utility function, and $\mathbf{A}$,
the parameter set of Pareto allocations, are defined in~\eqref{eq:8}
and~\eqref{eq:12}, and
\begin{equation*}
  \mathbf{B} \set (-\infty,0)^M \times (0,\infty)
  \times \mathbf{R}^J.
\end{equation*}
These functions represent the initial values of the stochastic fields
$F = (F_0,F_1)$ and $G=(G_0,G_1)$; the terminal values are given by
\begin{align*}
  F_1(a) &\set r(v,\Sigma(x,q)), \quad
  a=(v,x,q) \in \mathbf{A}, \\
  G_1(b) &\set \sup_{v\in (0,\infty)^M}\inf_{x\in
    \mathbf{R}}[\ip{v}{u} + xy - F_1(v,x,q)], \quad b = (u,y,q)\in
  \mathbf{B}.
\end{align*}
The stochastic fields $F$ and $G$ are thoroughly studied in our
companion paper \cite{BankKram:13b} under Assumption~\ref{as:1} and
the condition:
\begin{displaymath}
  F_0(a) =  \mathbb{E}[r(v,\Sigma(x,q))] > -\infty, \quad a=(v,x,q)
  \in \mathbf{A}, 
\end{displaymath}
which, by Lemma~\ref{lem:1}, is equivalent to Assumption~\ref{as:3}.

Theorem~4.1 in \cite{BankKram:13b} describes in detail the properties
of the sample paths of the stochastic fields $F$ and $G$. In
particular, the functions $F_0$ and $G_0$ are continuously
differentiable, and the derivatives of $F_0$ can be computed by taking
expectations of the corresponding derivatives of $F_1$:
\begin{equation}
  \label{eq:20}
  \frac{\partial F_0}{\partial a^i}(a) = 
  \mathbb{E}[\frac{\partial F_1}{\partial a^i}(a)], \quad
  a \in \mathbf{A}.
\end{equation}
Moreover, Theorem~2.2 in \cite{BankKram:13b} shows that, for every
fixed $q\in \mathbf{R}^J$, the following conjugacy relationships
between $(v,x)\in (0,\infty)^M\times\mathbf{R}$ and $(u,y)\in
(-\infty,0)^M\times (0,\infty)$ are equivalent:
\begin{enumerate}
\item \label{item:3} We have $x = \frac{\partial G_0}{\partial
    y}(u,y,q) = G_0(u,1,q)$ and $v = \frac{\partial G_0}{\partial
    u}(u,y,q)$.
\item \label{item:4} We have $y = \frac{\partial F_0}{\partial
    x}(v,x,q)$ and $u = \frac{\partial F_0}{\partial v}(v,x,q)$.
\end{enumerate}
In addition, in this case,
\begin{equation}
  \label{eq:21}
  \frac{\partial G_0}{\partial q}(u,y,q) = - \frac{\partial F_0}{\partial
    q}(v,x,q).
\end{equation} 

For $a\in \mathbf{A}$, recall the notations $\pi(a)$ and
$\mathbb{Q}(a)$ for the Pareto allocation and the pricing measure
from~\eqref{eq:14} and~\eqref{eq:16}, respectively. Also recall that
$\mathbf{S}^M$ denotes the interior of the simplex in $\mathbf{R}^M$.
The following result improves Theorem~\ref{th:1}.

\begin{Theorem}
  \label{th:3}
  Under Assumptions~\ref{as:1} and~\ref{as:3}, for every position
  $q\in \mathbf{R}^J$ there is a unique cash amount $x(q)$ and a
  unique Pareto optimal allocation $\alpha_1(q)$ with total endowment
  $\Sigma(x(q),q)$ such that
  \begin{equation}
    \label{eq:22}
    U^m_0 \set \mathbb{E}[u_m(\alpha^m_0)] = \mathbb{E}[u_m(\alpha^m_1(q))],
    \quad m=1,\dots, M. 
  \end{equation}
  The Pareto optimal allocation $\alpha_1(q)$ has the form
  \begin{equation}
    \label{eq:23}
    \alpha_1(q) = \pi(a(q)),\quad a(q)=(w(q),x(q),q),
  \end{equation}
  where the weights $w(q)\in \mathbf{S}^M$ and the cash amount
  $x(q)\in \mathbf{R}$ are given by
  \begin{align}
    \label{eq:24}
    w^m(q) &= \frac{\partial G_0}{\partial u^m}(U_0,1,q)/\sum_{k=1}^M
    \frac{\partial G_0}{\partial u^k}(U_0,1,q) , \quad m=1,\dots,M. \\
    \label{eq:25}
    x(q) & = G_0(U_0,1,q).
  \end{align}
  The function $\map{x=x(q)}{\mathbf{R}^J}{\mathbf{R}}$ is convex,
  continuously differentiable, and, for $q\in \mathbf{R}^J$,
  \begin{equation}
    \label{eq:26}
    \frac{\partial x}{\partial q^j}(q) =
    -\mathbb{E}_{\mathbb{Q}(a(q))}[\psi^j], \quad j=1,\dots,M.  
  \end{equation}
\end{Theorem}

\begin{proof}
  The uniqueness of the cash amount and of the Pareto optimal
  allocation with the desired properties follows directly from the
  definition of Pareto optimality and the strict concavity of utility
  functions.

  For the existence, consider the cash amount $x(q)$, the weights
  $w(q)$, and the Pareto optimal allocation $\alpha_1(q)$ defined
  by~\eqref{eq:23}--\eqref{eq:25}. Clearly, by the construction of the
  random field $\pi=\pi(a)$, the total endowment of $\alpha_1(q)$
  equals $\Sigma(x(q),q)$. Denoting
  \begin{displaymath}
    v(q) \set \frac{\partial G_0}{\partial u^m}(U_0,1,q)
    \quad\text{and} \quad \widetilde a(q) \set (v(q),x(q),q), 
  \end{displaymath}
  we deduce from the equivalence of items~\ref{item:3}
  and~\ref{item:4} above that
  \begin{align*}
    U^m_0 &= \frac{\partial F_0}{\partial v^m}(\widetilde a(q)), \quad
    m=1,\dots,M, \\
    1 & = \frac{\partial F_0}{\partial x}(\widetilde a(q)),
  \end{align*}
  and then from~\eqref{eq:21} that
  \begin{displaymath}
    \frac{\partial G_0}{\partial q^j}(U_0,1,q)=-\frac{\partial F_0}{\partial
      q^j}(\widetilde a(q)), \quad j=1,\dots,J.
  \end{displaymath}
  Accounting for~\eqref{eq:20} and the constructions of $\pi(a)$ and
  $\mathbb{Q}(a)$ in~\eqref{eq:15} and~\eqref{eq:16} we obtain
  \begin{align*}
    U^m_0 &=\mathbb{E} [\frac{\partial r}{\partial
      v^m}(v(q),\Sigma(x(q),q))] =
    \mathbb{E} [u_m(\pi^m(\widetilde a(q))],  \\
    1 & =
    \mathbb{E}[\frac{\partial r}{\partial x}(v(q),\Sigma(x(q),q)) ], \\
    \frac{\partial G_0}{\partial q^j}(U_0,1,q) & = -
    \mathbb{E}[\frac{\partial r}{\partial
      x}(v(q),\Sigma(x(q),q))\psi^j ] =
    -\mathbb{E}_{\mathbb{Q}(\widetilde a(q))}[\psi^j].
  \end{align*}
  As $w(q)$ is a scalar multiple of $v(q)$, we have $\pi(a(q)) =
  \pi(\widetilde a(q))$ and $\mathbb{Q}_{a(q)} =
  \mathbb{Q}_{\widetilde a(q)}$ and the equalities \eqref{eq:22} and
  \eqref{eq:26} follow.

  Finally, the convexity and the continuous differentiability of
  $x=x(q)$ follow from~\eqref{eq:25} and the properties of $G_0$ given
  in Theorem~4.1 in \cite{BankKram:13b}.
\end{proof}

Hereafter, we denote by $\map{x=x(q)}{\mathbf{R}^J}{\mathbf{R}}$ and
$\map{w=w(q)}{\mathbf{R}^J}{\mathbf{S}^M}$ the market indifference
price and the weights of the post-trade Pareto allocation defined by
\eqref{eq:25} and \eqref{eq:24}, respectively.  The sensitivity
analysis of these maps is accomplished in Theorem~\ref{th:4} below,
whose formulation requires us to introduce extra notations.

From now on, we also work under Assumption~\ref{as:2}.  We define the
market makers' terminal {risk tolerances}
\begin{displaymath}
  \tau^m(a) \set \frac1{a_m(\pi^m(a))} = -
  \frac{u'_m(\pi^m(a))}{u''_m(\pi^m(a))}, \quad m=1,\dots,M,\; a \in
  \mathbf{A},  
\end{displaymath}
and the {aggregate terminal risk-tolerance}
\begin{displaymath}
  R_1(a) \set \sum_{m=1}^M \tau^m(a), \quad a=(v,x,q) \in
  \mathbf{A}.
\end{displaymath}
Taking constant $c>0$ from Assumption~\ref{as:2} we have
\begin{displaymath}
  \frac1c \leq \tau^m(a) \leq c, \quad m=1,\dots,M,\; a \in
  \mathbf{A},
\end{displaymath}
and then
\begin{displaymath}
  \frac{M}c \leq R_1(a) \leq Mc, \quad a \in \mathbf{A}. 
\end{displaymath}
The latter estimate allows us to define a probability measure
$\mathbb{R}(a)$ whose density under $\mathbb{Q}(a)$ is given by
\begin{equation}
  \label{eq:27}
  \frac{d\mathbb{R}(a)}{d\mathbb{Q}(a)} \set
  \frac{1/R_1(a)}{\mathbb{E}_{\mathbb{Q}(a)}[1/R_1(a)]}, \quad
  a \in \mathbf{A},
\end{equation}
as well as the initial aggregate risk-tolerance:
\begin{displaymath}
  R_0(a) \set \mathbb{E}_{\mathbb{R}(a)}[R_1(a)], \quad a\in
  \mathbf{A}. 
\end{displaymath}
The terminology is due to the fact that $R=(R_0,R_1)$ is the
stochastic field of risk-tolerances of the stochastic field
$F=(F_0,F_1)$ of aggregate utilities:
\begin{equation}
  \label{eq:28}
  R_i(a) = -\frac{\partial F_i}{\partial x}(a)/\frac{\partial^2
    F_i}{\partial x^2}(a), \quad i=1,2;  
\end{equation}
see the discussion in \cite{BankKram:13b} just before Lemma~4.8.

Finally, for $a\in \mathbf{A}$, we define the matrices
\begin{align}
  \label{eq:29}
  A^{lm}(a) &=
  \frac1{R_0(a)}\left\{\delta_{lm}\mathbb{E}_{\mathbb{R}(a)}[\tau^m(a)
    R_1(a)]-\Cov{\mathbb{R}(a)}{\tau^l(a)}{\tau^m(a)}\right\},
  \\
  \label{eq:30}
  C^{mj}(a) &= \frac1{R_0(a)} \Cov{\mathbb{R}(a)}{\tau^m(a)}{\psi^j},
  \\
  \label{eq:31}
  D^{ij}(a) &= \frac1{R_0(a)} \Cov{\mathbb{R}(a)}{\psi^i}{\psi^j},
\end{align}
where $\delta_{lm}=\ind{l=m}$ is Kronecker's delta, we let
$l,m=1,\dots,M$ and $i,j=1,\dots, J$, and where
\begin{displaymath}
  \Cov{\mathbb{R}}{\xi}{\eta} \set \mathbb{E}_{\mathbb{R}}[\xi\eta] -
  \mathbb{E}_{\mathbb{R}}[\xi]\mathbb{E}_{\mathbb{R}}[\eta], 
\end{displaymath}
denotes the covariance of random variables $\xi$ and $\eta$ under a
probability measure $\mathbb{R}$. Lemma~4.8 in \cite{BankKram:13b}
shows that $A(a)$ is strictly positive definite:
\begin{displaymath}
  \frac1{c} \abs{z}^2 \leq \ip{z}{A(a)z} \leq  c\abs{z}^2, \quad
  z \in \mathbf{R}^M,
\end{displaymath}
where $c>0$ is the constant from Assumption~\ref{as:2}. In particular,
the inverse matrix $A(a)^{-1}$ is well-defined.

\begin{Theorem}
  \label{th:4}
  Let Assumptions~\ref{as:1},~\ref{as:2}, and~\ref{as:3} hold.  Then
  the Pareto weights $w=w(q)$ are continuously differentiable, the
  market indifference price $x=x(q)$ is two-times continuously
  differentiable and, for $m=1,\dots,M$, $i,j=1,\dots,J$, and $q\in
  \mathbf{R}^J$,
  \begin{align}
    \label{eq:32}
    Z^{mj}(q) &\set \frac1{w^m}\frac{\partial w^m}{\partial q^j}(q) =
    E^{mj}(q) -
    \sum_{k=1}^M w^k(q) E^{kj}(q), \\
    \label{eq:33}
    \frac{\partial^2 x}{\partial q^i \partial q^j}(q) &= H^{ij}(q),
  \end{align}
  where the matrices $E=E(q)$ and $H = H(q)$ have the expressions
  \begin{align}
    \label{eq:34}
    E &= -A^{-1}C, \\
    \label{eq:35}
    H &= C^T A^{-1}C + D,
  \end{align}
  in terms of the matrices $A\set A(a(q))$, $C\set C(a(q))$, and
  $D\set D(a(q))$ computed in \eqref{eq:29}, \eqref{eq:30}, and
  \eqref{eq:31} at $a(q)\set(w(q),x(q),q)$.
\end{Theorem}

The proof relies on the expressions below for the matrices $A(a)$,
$C(a)$, and $D(a)$ in terms of the partial derivatives of $F_0$. In
the case of the matrix $A(a)$, such a formula is given in Lemma~4.8 of
\cite{BankKram:13b}:
\begin{displaymath}
  A^{lm}(v,x,q) = \frac{v^lv^m}{\frac{\partial F_0}{\partial x}}
  \left(\frac{\partial^2 F_0}{\partial v^l\partial 
      v^m} -  \frac1{\frac{\partial^2 F_0}{\partial x^2}} \frac{\partial^2
      F_0}{\partial v^l\partial 
      x}\frac{\partial^2 F_0}{\partial v^m\partial
      x}\right)(v,x,q). 
\end{displaymath}
The corresponding computations for the matrices $C(a)$ and $D(a)$ are
done in the following lemma.

\begin{Lemma}
  \label{lem:4}
  Under Assumptions~\ref{as:1}, \ref{as:2}, and~\ref{as:3}, the
  matrices $C(a)$ and $D(a)$ defined by~\eqref{eq:30}
  and~\eqref{eq:31} can be written as
  \begin{align*}
    C^{mj}(v,x,q) &= \frac{v^m}{\frac{\partial F_0}{\partial x}}
    \left(\frac{\partial^2 F_0}{\partial v^m \partial q^j} -
      \frac1{\frac{\partial^2 F_0}{\partial x^2}} \frac{\partial^2
        F_0}{\partial v^m\partial x} \frac{\partial^2 F_0}{\partial
        x \partial
        q^j}\right)(v,x,q), \\
    D^{ij}(v,x,q) &= \frac{1}{\frac{\partial F_0}{\partial x}}
    \left(-\frac{\partial^2 F_0}{\partial q^i \partial q^j} +
      \frac1{\frac{\partial^2 F_0}{\partial x^2}} \frac{\partial^2
        F_0}{\partial x\partial q^i} \frac{\partial^2 F_0}{\partial
        x\partial q^j}\right)(v,x,q).
  \end{align*}
\end{Lemma}
\begin{proof}
  From the identities for the second derivatives of $r=r(v,x)$
  collected in Theorem~3.2 of~\cite{BankKram:13b} we deduce the
  following expressions for second order derivatives of $F_1(a) =
  r(v,\Sigma(x,q))$:
  \begin{align*}
    v^m\frac{\partial^2 F_1}{\partial v^m\partial x}(a) &=
    -\frac{\partial^2
      F_1}{\partial x^2}(a) \tau^m(a), \\
    v^m\frac{\partial^2 F_1}{\partial v^m\partial q^j}(a) &=
    -\frac{\partial^2
      F_1}{\partial x^2}(a) \tau^m(a)\psi^j, \\
    \frac{\partial^2 F_1}{\partial x\partial q^j}(a) &=
    \frac{\partial^2
      F_1}{\partial x^2}(a) \psi^j, \\
    \frac{\partial^2 F_1}{\partial q^i\partial q^j}(a) &=
    \frac{\partial^2 F_1}{\partial x^2}(a) \psi^i\psi^j.
  \end{align*}
  By Theorem 4.2 in~\cite{BankKram:13b}, the function $F_0 = F_0(a)$
  is twice continuously differentiable and its second order
  derivatives are expectations of the corresponding derivatives for
  $F_1 = F_1(a)$:
  \begin{displaymath}
    \frac{\partial^2 F_0}{\partial a^i\partial a^j}(a) =
    \mathbb{E}[\frac{\partial^2 
      F_1}{\partial a^i\partial a^j}(a)].  
  \end{displaymath} 
  The result now follows by direct computations if we account for the
  identities~\eqref{eq:28} for $R_i(a)$, $i=1,2$ and the expression
  \begin{displaymath}
    \frac{d\mathbb{R}(a)}{d\mathbb{P}} \set \frac{\partial^2 F_1}{\partial
      x^2}(a)/\frac{\partial^2 F_0}{\partial x^2}(a),
  \end{displaymath}
  for the density of $\mathbb{R}(a)$ under $\mathbb{P}$.
\end{proof}

\begin{proof}[Proof of Theorem~\ref{th:4}.]
  Theorem~4.2 in \cite{BankKram:13b} implies that $G_0 = G_0(a)$ is
  twice continuously differentiable. In view of the
  representations~\eqref{eq:24}, for $w(q)$, and~\eqref{eq:25}, for
  $x(q)$, we obtain that $w=w(q)$ is continuously differentiable and
  $x=x(q)$ is twice continuously differentiable. We also deduce that
  the identities \eqref{eq:32} and~\eqref{eq:33} hold with the
  matrices $E(q)$ and $H(q)$ such that
  \begin{align*}
    E^{mj}(q) &= \frac1{\frac{\partial G_0}{\partial u^m}(U_0,1,q)}
    \frac{\partial^2 G_0}{\partial u^m \partial q^j}(U_0,1,q), \\
    H^{ij}(q) &= \frac{\partial^2 G_0}{\partial q^i \partial
      q^j}(U_0,1,q),
  \end{align*}
  where $m=1,\dots,M$ and $i,j=1,\dots,J$.  The fact that these
  matrices can be computed by~\eqref{eq:34} and~\eqref{eq:35} is
  proved in Theorems~4.2 and 2.10 of \cite{BankKram:13b} in the
  situation where the matrices $A(a)$, $C(a)$, and $D(a)$ have the
  above expressions in terms of $F_0$.
\end{proof}

\subsection{Asymptotic expansions for market indifference prices}
\label{sec:asyptotic_expansion}

To facilitate the economic interpretation of the sensitivities given
in Theorem~\ref{th:3} let us compute the second-order expansion for
the market indifference prices $x=x(q)$ with respect to the order size
$q$.  Some additional notations are needed. For a random variable
$\xi$ and a probability measure $\mathbb{R}$ denote
\begin{displaymath}
  \Var{\mathbb{R}}{\xi} \set \Cov{\mathbb{R}}{\xi}{\xi},
\end{displaymath}
the variance of $\xi$ under $\mathbb{R}$.  For vectors $\mu\in
\mathbf{S}^M$ and $z\in \mathbf{R}^M$ we shall use a similar notation:
\begin{displaymath}
  \Var{\mu}{z} \set \sum_{i=1}^M \mu^i (z^i)^2 - (\sum_{i=1}^M \mu^i
  z^i)^2 
\end{displaymath}
interpreted as the variance of the random variable $z$ defined on the
sample space $\{1,\dots,M\}$ with respect to the probability measure
$\mu$.  For $a\in \mathbf{A}$ define $\rho(a)\in
\mathbf{L}^0(\mathbf{S}^M)$, the vector of relative risk-tolerance
weights, by
\begin{equation}
  \label{eq:36}
  \rho^m(a) \set \frac{\tau^m(a)}{\sum_{k=1}^M \tau^k(a)} =
  \frac{\tau^m(a)}{R_1(a)}, \quad m=1,\dots, M. 
\end{equation}

\begin{Theorem}
  \label{th:5}
  Under Assumptions~\ref{as:1}, \ref{as:2}, and \ref{as:3}, the
  second-order expansion for the market indifference prices $x=x(q)$
  with respect to the order $q\in \mathbf{R}^J$ can be written as
  \begin{equation}
    \label{eq:37}
    \begin{split}
      x(q+\Delta q) - x(q) &= -\mathbb{E}_{\mathbb{Q}}[\ip{\Delta
        q}{\psi}] + \frac{R_0}2
      \mathbb{E}_{\mathbb{R}}\left[\left(\frac{d\mathbb{Q}}{d\mathbb{R}}\right)^2
        \Var{\rho}{Z\Delta q}\right]
      \\
      &+\frac1{2R_0}\left\{
        \left(\cov_{\mathbb{R}}[\frac{d\mathbb{Q}}{d\mathbb{R}},\ip{\Delta
            q}{\psi}]\right)^2
        + \Var{\mathbb{R}}{\ip{\Delta q}{\psi}}\right\} \\
      & + o(\abs{\Delta q}^2), \quad \Delta q\to 0,
    \end{split}
  \end{equation}
  where $Z=Z(q)$ is defined in \eqref{eq:32} and where we omitted the
  argument $a(q) = (w(q),x(q),q)$.
\end{Theorem}

\begin{proof}
  The linear term in \eqref{eq:37} follows from \eqref{eq:26}. To
  verify the second-order part we decompose the matrix $A=A(a(q))$
  from~\eqref{eq:29} as
  \begin{displaymath}
    A = {R_0}(S + ss^T),
  \end{displaymath}
  where, for $l,m=1,\dots,M$,
  \begin{align*}
    S^{lm} &\set \frac1{R_0^2}
    \mathbb{E}_{\mathbb{R}}[\tau^l(\delta_{lm}\sum_{k=1}^M
    \tau^k - \tau^m)], \\
    s^l &\set \frac1{R_0}\mathbb{E}_{\mathbb{R}}[\tau^l].
  \end{align*}
  Denote $\idvec\set(1,\dots,1)\in \mathbf{R}^M$ and observe that
  \begin{align}
    \label{eq:38}
    S\idvec &= (\idvec^T S)^T = 0, \\
    \label{eq:39}
    \ip{s}{\idvec} &= \frac1{R_0}\mathbb{E}_{\mathbb{R}}[\sum_{k=1}^M
    \tau^k] = \frac1{R_0}\mathbb{E}_{\mathbb{R}}[R_1] = 1.
  \end{align}

  Let $E=E(q)$ and $H=H(q)$ be the matrices defined by~\eqref{eq:34}
  and~\eqref{eq:35}. Using the fact that $A=A^T$ we can write $H$ as
\begin{displaymath}
  H = E^T A E + D. 
\end{displaymath}
By Theorem~\ref{th:4}, $H=H(q)$ is the Hessian matrix for
$x=x(q)$. Hence, the second-order part in the expansion of $x(q+\Delta
q)-x(q)$ with respect to $\Delta q$ is given by
  \begin{align*}
    \frac12\ip{\Delta q}{H\Delta q} &= \frac12\ip{E\Delta q}{AE\Delta
      q} + \frac12\ip{\Delta q}{D\Delta q} \\ &=
    \frac{R_0}2\ip{E\Delta q}{SE\Delta q} +\frac{R_0}2(\ip{E\Delta
      q}{s})^2 + \frac12\ip{\Delta q}{D\Delta q}.
  \end{align*}
  Since, by \eqref{eq:32}, $(E-Z)\Delta q$ is the product of some
  scalar on the vector $\idvec$, \eqref{eq:38} implies that
  $\ip{E\Delta q}{SE\Delta q} = \ip{Z\Delta q}{SZ\Delta q}$.  Observe
  that, in view of \eqref{eq:27}, the matrix $S$ can be written as
  \begin{align*}
    S^{lm} = \frac1{R_0^2} \mathbb{E}_{\mathbb{R}}[R^2_1
    \rho^l(\delta_{lm} - \rho^m)] =
    \mathbb{E}_{\mathbb{R}}[\left(\frac{d\mathbb{Q}}{d\mathbb{R}}\right)^2
    \rho^l(\delta_{lm} - \rho^m)].
  \end{align*}
  It follows that
  \begin{displaymath}
    \ip{Z\Delta q}{SZ\Delta q} =
    \mathbb{E}_{\mathbb{R}}[\left(\frac{d\mathbb{Q}}{d\mathbb{R}}\right)^2
    \Var{\rho}{Z\Delta q}], 
  \end{displaymath}
  giving the first quadratic term in \eqref{eq:37}.

  By \eqref{eq:34}, $\ip{\idvec}{(AE + C)\Delta q} = 0$, which, in
  view of \eqref{eq:38} and \eqref{eq:39}, implies that
  \begin{displaymath}
    R_0\ip{s}{E\Delta q} + \ip{\idvec}{C\Delta q} = 0.
  \end{displaymath}
  From the construction of $C$ in \eqref{eq:30} and accounting again
  for \eqref{eq:27} we deduce the second quadratic term in
  \eqref{eq:37}:
  \begin{align*}
    \ip{s}{E\Delta q} &= - \frac1{R_0} \ip{\idvec}{C\Delta q} =
    -\frac1{R^2_0} \cov_{\mathbb{R}}[R_1,\ip{\Delta q}{\psi}] \\
    &= -
    \frac1{R_0}\cov_{\mathbb{R}}[\frac{d\mathbb{Q}}{d\mathbb{R}},\ip{\Delta
      q}{\psi}].
  \end{align*}
 
  Finally, the expression \eqref{eq:31} for $D$ yields the last term:
  \begin{displaymath}
    \ip{\Delta q}{D\Delta q} = \frac1{R_0}\Var{\mathbb{R}}{\ip{\Delta q}{\psi}}.
  \end{displaymath}
\end{proof}

Observe that the linear term in~\eqref{eq:37} corresponds to the
``standard'' model of mathematical finance, where a \emph{``small''
  investor} can trade \emph{any number} of securities $\psi$ at
``fixed'' or \emph{exogenous} prices $\mathbb{E}_{\mathbb{Q}}[\psi]$.
The second, quadratic, component can thus be viewed as a \emph{price
  impact correction} to this model.  Note that all three terms of the
quadratic part are \emph{non-negative} and the last term,
$\Var{\mathbb{R}}{\ip{\Delta q}{\psi}}$, equals zero iff $\ip{\Delta
  q}{\psi}=\const$. Hence, for any non-trivial transaction our large
investor will have to pay a \emph{strictly positive penalty} due to
his price impact in comparison with a hypothetical small agent trading
at $\mathbb{E}_{\mathbb{Q}}[\psi]$.

To get a further insight into the microstructure of our model we
provide still another second-order expansion of $x=x(q)$ along with
the first-order expansion of the post-trade Pareto allocations $\pi =
\pi(q) \set \pi(a(q))$.  We introduce probability measures
$\mathbb{R}^m(a)$, $m=1,\dots,M$, such that
\begin{displaymath}
  \frac{d\mathbb{R}^m(a)}{d\mathbb{Q}(a)} \set
  \frac{1/\tau^m(a)}{\mathbb{E}_{\mathbb{Q}(a)}[1/\tau^m(a)]}, \quad 
  m=1,\dots,M, \; a \in \mathbf{A};
\end{displaymath}
they are analogs of $\mathbb{R}(a)$ for individual market makers.  We
shall also make use of the random matrix $\zeta = \zeta(a)$ given, for
$a \in \mathbf{A}$, by
\begin{equation}
  \label{eq:40}
  \zeta^{mj}(a) \set 
  R_1(a) \rho^m(a)\left(E^{mj}(a) - \sum_{l=1}^M \rho^l(a)
    E^{lj}(a)\right),
\end{equation}
where $ m=1,\dots,M$, $j=1,\dots,J$, and the matrix $E=E(a)$ is
defined in~\eqref{eq:34}. We denote by $\zeta^m$ the $m$th row of
$\zeta$ and observe that
\begin{equation}
  \label{eq:41}
  \sum_{m=1}^M \zeta^{m} = 0.
\end{equation}

\begin{Theorem}
  \label{th:6}
  Under Assumptions~\ref{as:1}, \ref{as:2}, and \ref{as:3}, the
  first-order expansion, in the almost sure sense, of the Pareto
  optimal allocations $\pi = \pi(q) \set \pi(a(q))$ is given by
  \begin{equation}
    \label{eq:42}
    \begin{split}
      \pi^m(q+\Delta q) - \pi(q) =& \rho^m \left(\ip{\psi}{\Delta q} -
        \mathbb{E}_{\mathbb{Q}}[\ip{\psi}{\Delta q}]\right) +
      \ip{\zeta^m}{\Delta q} \\
      & + o(\abs{\Delta q}^2), \quad \Delta q\to 0, \; m=1,\dots,M,
    \end{split}
  \end{equation} 
  and the second-order expansion for the market indifference prices
  $x=x(q)$ can be written as
  \begin{equation}
    \label{eq:43}
    \begin{split}
      x(q+\Delta q) - x(q) =& -\mathbb{E}_{\mathbb{Q}}[\ip{\Delta
        q}{\psi}]
      \\
      &+\frac12\Bigg\{\mathbb{E}_{\mathbb{Q}}[1/R_1]\mathbb{E}_{\mathbb{R}}
      \left[\ip{\psi-\mathbb{E}_{\mathbb{Q}}[\psi]}{\Delta
          q}^2\right]\\
      &\qquad+\sum_{m=1}^M \mathbb{E}_{\mathbb{Q}}[1/\tau^m]
      \mathbb{E}_{\mathbb{R}^m}\left[\ip{\zeta^m}{\Delta
          q}^2\right]\Bigg\}\\
      & + o(\abs{\Delta q}^2), \quad \Delta q\to 0,
    \end{split}
  \end{equation}
  where we omitted the argument $a(q) = (w(q),x(q),q)$.
\end{Theorem}
\begin{proof}
  To verify~\eqref{eq:42} recall that if $\widehat x^m = \widehat
  x^m(v,x)$, $m=1,\dots,M$, is the argmax vector in the construction
  of the aggregate utility:
  \begin{displaymath}
    r(v,x) \set \sup_{x^1 + \dots + x^M = x} \sum_{m=1}^M v^m u_m(x^m) =
    \sum_{m=1}^M v^m u_m(\widehat x^m),
  \end{displaymath}
  then its partial derivatives are given by
  \begin{align*}
    \frac{\partial \widehat x^m}{\partial x}(v,x) & =
    \frac{t_m(\widehat
      x^m)}{\sum_{k=1}^M t_k(\widehat x^k)}, \\
    v^l \frac{\partial \widehat x^m}{\partial v^l}(v,x) & = v^m
    \frac{\partial \widehat x^l}{\partial v^m}(v,x) = t_m(\widehat
    x^m)\left(\delta_{lm} - \frac{t_l(\widehat x^l)}{\sum_{k=1}^M
        t_k(\widehat x^k)}\right),
  \end{align*}
  see, e.g., Theorem~3.2 in \cite{BankKram:13b}, where
  \begin{displaymath}
    t_m(x) \set \frac1{a_m(x)} =
    -\frac{u'_m(x)}{u''_m(x)}, \quad x\in \mathbf{R}, \; m=1,\dots,M, 
  \end{displaymath}
  are the market makers' risk-tolerances. It follows that, for
  $a=(v,x,q)$,
  \begin{align*}
    \frac{\partial \pi^m}{\partial x}(a) & = \rho^m(a), \\
    \frac{\partial \pi^m}{\partial q^j}(a) & = \rho^m(a) \psi^j,
    \\
    v^l\frac{\partial \pi^m}{\partial v^l}(a) & = R_1(a) \rho^m(a)
    \left(\delta_{lm} - \rho^l(a)\right).
  \end{align*}
  Combining these identities with~\eqref{eq:26} and~\eqref{eq:32} and
  omitting, as usual, the argument $a(q) = (w(q),x(q),q)$ we obtain
  \begin{align*}
    \frac{\partial }{\partial q^j} \pi^m(q) &= \rho^m \frac{\partial
      x(q)}{\partial q^j} + \rho^m \psi^j + R_1(a)\rho^m(a)
    \sum_{l=1}^M \left(\delta_{lm} - \rho^l(a)\right)
    \frac1{w^l(q)}\frac{\partial w^l(q)}{\partial
      q^j} \\
    &= \rho^m (-\mathbb{E}_{\mathbb{Q}}[\psi^j]) + \rho^m \psi^j  \\
    &\quad + R_1(a)\rho^m(a) \sum_{l=1}^M \left(\delta_{lm} -
      \rho^l(a)\right) \left(E^{lj}(q) - \sum_{k=1}^M
      w^k(q)E^{kj}(q)\right) \\
    &= \rho^m (\psi^j-\mathbb{E}_{\mathbb{Q}}[\psi^j]) + \zeta^{mj},
  \end{align*}
  which yields~\eqref{eq:42}.

  To verify the second-order expansion~\eqref{eq:43} for $x=x(q)$ we
  shall match its terms with those of~\eqref{eq:37}. Recall that
  ${d\mathbb{Q}}/{d\mathbb{R}} = {R_1}/{R_0}$ and observe that for a
  random variable $\xi$
  \begin{align*}
    \mathbb{E}_{\mathbb{R}}[(\xi - \mathbb{E}_{\mathbb{Q}}[\xi])^2 &=
    \mathbb{E}_{\mathbb{R}}[(\xi - \mathbb{E}_{\mathbb{R}}[\xi])^2 +
    (\mathbb{E}_{\mathbb{R}}[\xi] -
    \mathbb{E}_{\mathbb{Q}}[\xi])^2 \\
    &= \Var{\mathbb{R}}{\xi} +
    \left(\cov_{\mathbb{R}}[\frac{d\mathbb{Q}}{d\mathbb{R}},\xi]\right)^2.
  \end{align*}
  This shows that the second term in~\eqref{eq:43} coincides with the
  sum of the last two terms in~\eqref{eq:37}.

  If now $\eta = (\eta^m)_{m=1,\dots,M}$ is a vector of random
  variables such that $\ip{\rho}{\eta} = \sum_{m=1}^M \rho^m \eta^m =
  0$, then
  \begin{align*}
    &\sum_{m=1}^M \mathbb{E}_{\mathbb{Q}}[1/{\tau^m}]
    \mathbb{E}_{\mathbb{R}^m}[(\tau^m \eta^m)^2] = \sum_{m=1}^M
    \mathbb{E}_{\mathbb{Q}}[\tau^m (\eta^m)^2] \\
    \quad & = R_0 \sum_{m=1}^M
    \mathbb{E}_{\mathbb{R}}\left[\left(\frac{d\mathbb{Q}}{d\mathbb{R}}\right)^2
      \rho^m (\eta^m)^2\right] = R_0
    \mathbb{E}_{\mathbb{R}}\left[\left(\frac{d\mathbb{Q}}{d\mathbb{R}}\right)^2
      \Var{\rho}{\eta}\right].
  \end{align*}
  Choosing
  \begin{displaymath}
    \eta^m = \frac1{\tau^m}\ip{\zeta^m}{\Delta q}, \quad m=1,\dots,M, 
  \end{displaymath}
  we deduce from~\eqref{eq:41} that $\ip{\rho}{\eta} = 0$ and then
  from \eqref{eq:40} and~\eqref{eq:32} that
  \begin{displaymath}
    \eta = E\Delta q + z_1 \idvec = Z\Delta q + z_2 \idvec, 
  \end{displaymath}
  where $\idvec = \braces{1,\dots,1}$ and $z_1$ and $z_2$ are some
  constants. It follows that
  \begin{displaymath}
    \Var{\rho}{\eta} = \Var{\rho}{E\Delta q} = \Var{\rho}{Z\Delta q}, 
  \end{displaymath}
  proving the equality between the third term in~\eqref{eq:43} and the
  second term in~\eqref{eq:37}.
\end{proof}

The expansion~\eqref{eq:42} shows that the market makers share the
extra endowment $\ip{\psi-\mathbb{E}_{\mathbb{Q}}[\psi]}{\Delta q}$ in
proportion to their individual risk tolerances and, in addition, make
\emph{zero-sum} trades $\ip{\zeta^m}{\Delta q}$, $m=1,\dots,M$,
between themselves; see~\eqref{eq:41}. As the expansion~\eqref{eq:43}
reveals, the liquidity premium arises from both these components.

It is quite common in the economic literature to replace a collection
of economic agents with a single, \emph{representative}, agent whose
utility function is given by $r(w,\cdot)$, with $r=r(v,x)$
from~\eqref{eq:8}, for some \emph{fixed} weight $w\in
\mathbf{S}^M$. In our case, this simplification yields the linear term
in \eqref{eq:43} as well as the first component of the quadratic part,
but not the remaining ones. It is interesting to obtain conditions for
the second term of the quadratic part in~\eqref{eq:43} or,
equivalently, the second term in~\eqref{eq:37} to vanish, since, then,
the representative agent approximation leads to the identical
expression for the price impact coefficient as our original model with
many market makers. This is accomplished in the following lemma.

\begin{Lemma}
  \label{lem:5}
  Let the conditions of Theorem~\ref{th:5} hold and take $q,\Delta
  q\in \mathbf{R}^J$. Then the following assertions are equivalent:
  \begin{enumerate}
  \item \label{item:5}
    $\mathbb{E}_{\mathbb{R}^m}\left[\ip{\zeta^m}{\Delta
        q}^2\right]=0$, $m=1,\dots,M$.
  \item \label{item:6}
    $\mathbb{E}_{\mathbb{R}}\left[\left(\frac{d\mathbb{Q}}{d\mathbb{R}}\right)^2 
      \Var{\rho}{Z(q)\Delta q}\right] = 0$;
  \item \label{item:7} $Z(q)\Delta q = 0$;
  \item \label{item:8} $\mathbb{E}_{\mathbb{Q}^m}[\ip{\Delta q}{\psi}]
    = \mathbb{E}_{\mathbb{Q}}[\ip{\Delta q}{\psi}]$, $m=1,\dots,M$,
  \end{enumerate}
  where, for $a\in \mathbf{A}$,
  \begin{displaymath}
    \frac{d\mathbb{Q}^m(a)}{d\mathbb{Q}(a)} \set
    \frac{\rho^m(a)}{\mathbb{E}_{\mathbb{Q}}[\rho^m(a)]}, \quad m=1,\dots,M,
  \end{displaymath}
  and we omitted the argument $a(q) \set (w(q),x(q),q)$.
\end{Lemma}

\begin{proof} 
  \ref{item:5} $\Longleftrightarrow$ \ref{item:6}: This follows from
  the equality:
  \begin{displaymath}
    \sum_{m=1}^M \mathbb{E}_{\mathbb{Q}}[1/\tau^m]
    \mathbb{E}_{\mathbb{R}^m}\left[\ip{\zeta^m}{\Delta q}^2\right] = {R_0}
    \mathbb{E}_{\mathbb{R}}\left[\left(\frac{d\mathbb{Q}}{d\mathbb{R}}\right)^2
      \Var{\rho}{Z\Delta q}\right]
  \end{displaymath}
  which was part of the proof of Theorem~\ref{th:6}.

  \ref{item:6} $\Longleftrightarrow$ \ref{item:7}: Clearly,
  item~\ref{item:6} holds if and only if $\Var{\rho}{Z(q)\Delta q}=
  0$, which, in turn, is equivalent to $Z(q)\Delta q=y \idvec$ for
  some $y\in \mathbf{R}$. From the construction of the matrix $Z(q)$
  in~\eqref{eq:32} we deduce that $\ip{w(q)}{Z(q)\Delta q} = 0$, where
  the Pareto weights $w(q)$ take values in $\mathbf{S}^M$. It follows
  that $y=0$.

  \ref{item:7} $\Longleftrightarrow$ \ref{item:8}: Denote $\xi \set
  \ip{\Delta q}{\psi} \in \mathbf{L}^0(\mathbf{R})$. From the
  definition of the measures $\mathbb{Q}^m$, $m=1,\dots,M$, we deduce
  the equivalence of item~\ref{item:8} to
  \begin{equation}
    \label{eq:44}
    \mathbb{E}_{\mathbb{Q}^m}[\xi] =
    \mathbb{E}_{\mathbb{Q}^1}[\xi], \quad m=2,\dots,M. 
  \end{equation}
  From~\eqref{eq:32} we deduce that $Z(q)\Delta q=0$ if and only if $E
  \Delta q = y\idvec$ for some $y\in \mathbf{R}$, where the matrix
  $E=E(q)$ satisfies~\eqref{eq:34}. Hence, item~\ref{item:7} is
  equivalent to the existence of a constant $y\in \mathbf{R}$ such
  that
  \begin{equation}
    \label{eq:45}
    y A\idvec  + C\Delta q = 0,  
  \end{equation}
  From the expressions~\eqref{eq:29} and~\eqref{eq:30} for the
  matrices $A = A(a(q))$ and $C=C(a(q))$ we obtain
  \begin{align*}
    (A\idvec)^m & = \mathbb{E}_{\mathbb{R}}[\tau^m] = R_0
    \mathbb{E}_{\mathbb{Q}}[\rho^m], \\
    (C\Delta q)^m &= \frac1{R_0}\cov_{\mathbb{R}}[\tau^m,\xi] =
    \frac1{R_0}\left(\mathbb{E}_{\mathbb{R}}[\tau^m \xi] -
      \mathbb{E}_{\mathbb{R}}[\tau^m]\mathbb{E}_{\mathbb{R}}[\xi]\right) \\
    &= \mathbb{E}_{\mathbb{Q}}[\rho^m]
    \left(\mathbb{E}_{\mathbb{Q}^m}[\xi] -
      \mathbb{E}_{\mathbb{R}}[\xi]\right),\quad m=1,\dots,M,
  \end{align*}
  which clearly implies the equivalence of~\eqref{eq:44}
  and~\eqref{eq:45}.
\end{proof}

The condition of item~\ref{item:8} is clearly satisfied when the
random weights $\rho$ defined in~\eqref{eq:36} are deterministic. This
is the case, for instance, if all market makers have exponential
utilities: $u_m(x) = - \exp(-a_m x)$, with constant risk-aversion
$a_m>0$, $m=1,\dots,M$. Moreover, if the securities $\psi$ form a
\emph{complete} model in the sense that, jointly with the constant
security paying $1$ they span all random variables, than the validity
of item~\ref{item:8} for any $\Delta q \in \mathbf{R}^J$ is in fact
equivalent to $\rho$ being deterministic.

\bibliographystyle{plainnat} \bibliography{../bib/finance}

\end{document}